\documentclass[11pt]{article}
\usepackage{amsfonts}
\usepackage{amsmath,amssymb}
\usepackage{hyperref}
\hypersetup{colorlinks=true}
\usepackage{cleveref}
\usepackage{algorithm,algorithmicx,algpseudocode,graphicx}
\usepackage{amsthm}
\usepackage{thm-restate}
\usepackage{color}
\usepackage{fullpage}
\usepackage{authblk}
\usepackage[letterpaper]{geometry}
\usepackage{xspace}
\usepackage{comment}

\newtheorem{corollary}{Corollary}
\newtheorem{lemma}{Lemma}
\newtheorem{theorem}{Theorem}

\title{Polylogarithmic Approximation Algorithm for $k$-Connected Directed Steiner Tree on Quasi-Bipartite Graphs}

\author[1]{Chun-Hsiang Chan\thanks{kenhchan@umich.edu}}
\author[2]{Bundit Laekhanukit\thanks{bundit@sufe.edu.cn}}
\author[3]{Hao-Ting Wei\thanks{hw2738@columbia.edu}}
\author[4]{Yuhao Zhang\thanks{yhzhang2@cs.hku.hk}}
\affil[1]{Department of Computer Science, University of Michigan}
\affil[2]{Institute for Theoretical Computer Science, Shanghai University of Finance \& Economics \vspace{1mm}}
\affil[3]{Department of IEOR, Columbia University}
\affil[4]{Department of Computer Science, The University of Hong Kong}

\def\max{{\rm max}}
\def\min{{\rm min}}
\long\def\longdelete#1{}

\newcommand{\Value}{\mathsf{Val}\xspace}
\newcommand{\LP}{\mathsf{LP}\xspace}
\newcommand{\Halo}{\mathsf{Halo}\xspace}
\newcommand{\cost}{\mathsf{cost}\xspace}

\newcommand{\opt}{\mathsf{OPT}\xspace}
\newcommand{\dist}{\sigma}
\newcommand{\E}{\mathbb{E}}

\date{}

\renewcommand{\setminus}{-}

\alglanguage{pseudocode}
\algrenewcommand\algorithmicrequire{\textbf{Input}}
\algrenewcommand\algorithmicensure{\textbf{Output}}

\begin{document}
\maketitle

\begin{abstract}
In the $k$-Connected Directed Steiner Tree problem ($k$-DST), we are given a directed graph $G=(V,E)$ with edge (or vertex) costs, a root vertex $r$, a set of $q$ terminals $T\subseteq V\setminus\{r\}$, and a connectivity requirement $k>0$; the goal is to find a minimum-cost subgraph $H\subseteq G$ such that $H$ has $k$ edge-disjoint paths from the root $r$ to each terminal $t\in T$.
The $k$-DST problem is a natural generalization of the classical Directed Steiner Tree problem (DST) in the fault-tolerant setting in which the solution subgraph is required to have an $r,t$-path, for every terminal $t\in T$, even after removing $k-1$ vertices or edges. 
This paper studies the $k$-DST problem when an input graph is quasi-bipartite, i.e., there is no edge joining two non-terminal vertices.

The fault-tolerant variants of DST have been actively studied in the past decades; see, e.g., [Cheriyan et al., SODA'12 \& TALG], [Laekhanukit, SODA'14], [Laekhanukit, ICALP'16], [Grandoni-Laekhanukit, STOC'18]. Despite this, for $k>2$, the positive results were known only in special cases, e.g., directed acyclic graphs when $|T|+k$ is a constant or in a $\gamma$-shallow instances for constant $\gamma$. In this paper, we make progress toward devising approximation algorithms for $k$-DST. We extend the study of DST in quasi-bipartite graphs [Hibi-Fujito, Algorithmica; Friggstad et al., SWAT'16] to the fault-tolerant setting and present a polynomial-time $O(\log k \log q)$-approximation algorithm for $k$-DST in quasi-bipartite graphs, for arbitrary $k\geq 1$. Our result is based on the Halo-Set decomposition developed by Kortsarz and Nutov [STOC'04 \& SICOMP] and further developed in subsequent works, e.g., [Fakcharoenphol-Laekhanukit, STOC'08 \& SICOMP], [Nutov, SODA'09 \& Combinatorica], [Nutov, FOCS'09 \& TALG]. The main ingredient in our work is a non-trivial reduction from the problem of covering uncrossable families of subsets to the Set Cover problem, which can be seen as the generalization of the spider decomposition method in [Klein-Rav, IPCO'93 \& JAL; Nutov, APPROX'06 \& TCS]. 
\end{abstract}
\section{Introduction}
\label{sec:intro}

Designing a network that can operate under failure conditions is an important task for Computer Networking in both theory and practice. Many models have been proposed to capture this problem, giving rise to the area of survivable and fault-tolerant network design. In the past few decades, there have been intensive studies on the survivable network design problems; see, e.g., \cite{WilliamsonGMV95,GoemansGPSTW94,Jain01,FleischerJW06,ChuzhoyK12,Nutov12,GrandoniL17}. The case of link-failure is modeled by the {\em Edge-Connectivity Survivable Network Design} problem (EC-SNDP), which is shown to admit a $2$-approximation algorithm by Jain \cite{Jain01}. The case of node-failure is modeled by the {\em Vertex-Connectivity Survivable Network Design} problem (VC-SNDP), which is shown to admit a polylogarithmic approximation algorithm by Chuzhoy and Khanna \cite{ChuzhoyK12}. Nevertheless, most of the known algorithmic results pertain to only undirected graphs, where each link has no prespecified direction. In the directed case, only a few results are known as the general case of Survivable Network Design is at least as hard as the {\em Label-Cover} problem \cite{DodisK99}, which is believed to admit no sub-polynomial approximation algorithm \cite{Moshkovitz15,BellareGLR93}.

This paper studies the special case of the Survivable Network Design problem on directed graphs, namely the {\em $k$-Connected Directed Steiner Tree} problem ($k$-DST), which is also known as the {\em Directed Root $k$-Connectivity} In this problem, we are given an $n$-vertex directed graph $G=(V,E)$ with edge-costs $c:E\rightarrow\mathbb{R}_0^+$, a root vertex $r$, a set of $q$ terminals $T\subseteq V\setminus\{r\}$ and a connectivity requirement $k\in\mathbb{Z}^+$; the goal is to find a minimum-cost subgraph $H\subseteq G$ that has $k$ edge-disjoint\footnote{We define the problem here as an edge-connectivity problem; however, in directed graphs, edge-connectivity and vertex-connectivity variants are equivalent. In addition, the edge-cost and the vertex-cost variants are also equivalent.} $r,t$-paths for every terminal $t\in T$. 
This problem was mentioned in \cite{FeldmanKN09} and have been subsequently studied in \cite{CheriyanLNV14,Laekhanukit14,ChalermsookGL15,Laekhanukit16,GrandoniL17}. The only known non-trivial approximation algorithms for $k$-DST are for the case $k=2$ due to the work of Grandoni and Laekhanukit \cite{GrandoniL17}, and for the case of $\gamma$-shallow instances due to the work of Laekhanukit \cite{Laekhanukit16}. To the best of our knowledge, for $k \geq 3$, there were only a couple of positive results on $k$-DST: (1) Laekhanukit \cite{Laekhanukit16} devised an approximation algorithm whose the running-time and approximation ratios depend on the diameter of the optimal solution, and (2) Chalermsook, Grandoni and Laekhanukit \cite{ChalermsookGL15} devised a bi-criteria approximation algorithm for a special case of $k$-DST, namely the {\em $k$-Edge-Connected Group Steiner Tree} ($k$-GST), where the solution subgraph is guaranteed to be an $O(\log^2n\log k)$-approximate solution, whereas the connectivity is only guaranteed to be at least $\Omega(k/\log n)$.
Our focus is the case of $k$-DST where an input graph is {\em quasi-bipartite}, i.e., there is no edge joining any pair of non-terminal (Steiner) vertices, which generalizes the works of Hibi-Fujito \cite{HibiF16} and Friggstad-K\"{o}nemann-Shadravan \cite{FriggstadKS16} for the classical {\em directed Steiner tree} problem (the case $k=1$).

The main contribution of this paper is an $O(\log q\log k)$-approximation algorithm for $k$-DST on quasi-bipartite graphs, which runs in polynomial-time regardless of the structure of the optimal solution. Our result can be considered the first true polylogarithmic approximation algorithm whose running time is independent of the structure (i.e., diameter) of the optimal solution, albeit the algorithm is restricted to the class of quasi-bipartite graphs. Our technique is completely different from all the previous works \cite{GrandoniL17,Laekhanukit16,ChalermsookGL15}; all these results rely on the tree-rounding algorithm for the {\em Group Steiner Tree} problem by Garg, Konjevod and Ravi \cite{GargKR00}, and thus require either an LP whose support is a tree or a tree-embedding technique (e.g., R\"{a}cke's decomposition \cite{Racke08} as used in \cite{ChalermsookGL15}). Our algorithm, on the other hand, employs the {\em Halo-Set decomposition} devised by Kortsarz and Nutov \cite{KortsarzN05} and further developed in a series of works \cite{FakcharoenpholL12,CheriyanL13,Nutov12,Nutov14,Laekhanukit15-subset,Nutov12-subset}.
It is worth noting that the families of subsets decomposed from our algorithm are not uncrossable. We circumvent this difficulty by reducing the problem of {\em covering uncrossable families} to the {\em Set Cover} problem. Our algorithm can be seen as a generalization of the {\em spider decomposition} method developed by Klein-Ravi~\cite{KleinR95} and Nutov~\cite{Nutov10}.

Lastly, we remark that it was discussed in \cite{GrandoniL17} that the tree-embedding approach reaches the barrier as soon as $k>2$, and this holds even for quasi-bipartite graphs. Please see \Cref{app:bad-example} for discussions. While our algorithm exploits the structure of quasi-bipartite graphs, we hope that our technique using the Halo-Set decomposition would be an alternative method that sheds some light in developing approximation algorithms for the general case of $k$-DST for $k>2$.

\subsection{Related Works}
\label{sec:intro:related}

Directed Steiner tree has been a central problem in combinatorial and optimization. There have been a series of work studying this problem; see, e.g., \cite{Zelikovsky97,CharikarCCDGGL99,Rothovoss11,FriggstadKKLST14,GrandoniLL19,GhugeN18}. The best approximation ratio of $O(q^{\epsilon})$, for any $\epsilon<0$, in the regime of polynomial-time algorithms is known in the early work of Charikar~et~al.~\cite{CharikarCCDGGL99}\footnote{The same result can be obtained by applying the algorithm by Peleg and Kortsarz in \cite{KortsarzP97}}, which leads to an $O(\log^3 q)$-approximation algorithm that runs in quasi-polynomial-time. Very recently, Gradoni, Laekhanukit and Li \cite{GrandoniLL19} developed a framework that gives a quasi-polynomial-time $O(\log^2 q/\log\log q)$-approximation algorithm for the Directed Steiner Tree problem, and this approximation ratio is the best possible for quasi-polynomial-time algorithms, assuming the {\em Projection Games Conjecture} and $\mathrm{NP}\subseteq\bigcup_{\delta>0}\mathrm{ZPTIME}(2^{n^{\delta}})$.
The same approximation ratio was obtained in an independent work of Ghuge and Nagarajan  \cite{GhugeN18}.

The study of Steiner tree problems on quasi-bipartite graphs was initiated by Rajagopalan and Vazirani \cite{RajagopalanV99} in order to understand the bidirected-cut relaxation of the (undirected) {\em Steiner tree} problem. Since then the special case of quasi-bipartite graphs has played a central role in studying the Steiner tree problem; see, e.g., \cite{Rizzi03,ChakrabartyDV11,RobinsZ00,KonemannPT11,ByrkaGRS13,GoemansORZ12}.
For the case of directed graphs, Hibi-Fujito \cite{HibiF16} and Friggstad-K\"{o}nemann-Shadravan \cite{FriggstadKS16} independently discovered $O(\log n)$-approximation algorithms for the directed Steiner tree problem on quasi-bipartite graphs. Assuming $\mathrm{P}\neq\mathrm{NP}$, this matches to the lower bound of $(1-\epsilon)\ln n$, for any $\epsilon>0$, inherited from the Set Cover problem \cite{Feige98,DinurS14}. 

The generalization of the Steiner tree problem is known as the Survivable Network Design problem, which has been studied in both edge-connectivity \cite{WilliamsonGMV95,GoemansGPSTW94,Jain01}, vertex-connectivity \cite{ChuzhoyK12} and element-connectivity \cite{FleischerJW06} settings. The edge and element connectivity Survivable Network Design problems admit factor $2$ approximation algorithms via the iterative rounding method, while the vertex-connectivity variant admits no polylogarithmic approximation algorithm \cite{KortsarzKL04,ChakrabortyCK08,Laekhanukit14} unless $\mathrm{NP}\subseteq\mathrm{DTIME}(n^{\mathrm{polylog}(n)})$. To date, the best approximation ratio known for the Vertex-Connectivity Survivable Network problem is $O(k^3\log n)$ due to the work of Chuzhoy and Khanna \cite{ChuzhoyK12}.

In vertex-connectivity network design, one of the most common technique is the Halo-Set decomposition method, which has been developed in a series of works \cite{KortsarzN05,FakcharoenpholL12,CheriyanL13,Nutov14}. The main idea is to use the number of minimal deficient sets as a notion of progress. Here a deficient set is a subset of vertices that needs at least one incoming edge to satisfy the connectivity requirement. The minimal deficient sets in \cite{KortsarzN05,FakcharoenpholL12,CheriyanL13,Nutov14}, called cores, are independent and have only polynomial number, while the total number of deficient sets is exponential on the number of vertices. The families of deficient sets defined by these cores allow us to keep track of how many deficient sets remain in a solution subgraph. The early version of this method can be traced back to the seminal result of Frank \cite{Frank99a} and that of Frank and Jordan \cite{FrankJ99}; please see \cite{FrankJ2016graph} for reference therein.

 The spider decomposition method was introduced by Klein and Ravi \cite{KleinR95} to handle the {\em Vertex-Weighted Steiner Tree} problem. This technique gives a tight approximation result (up to constant factor) to the problem. Later, Nutov generalized the technique to deal with the {\em Minimum Power-Cover} problems \cite{Nutov10} and subsequently for the {\em Vertex-Weighted Element-Connectivity Survivable Network Design} problem \cite{Nutov12}.

\subsection {Technical Difficulties}
\label{sec:tech-difficult}

As mentioned, our algorithm relies much on the combination of the known techniques. However, due to structural differences, there are quite a few obstacles in adapting these techniques in our settings. We discuss in this section the structural differences, which might help the readers in studying $k$-DST.

\begin{itemize}
    \item {\bf No Tree Decomposition.} Firstly, as we mentioned the tree-embedding technique is not available for us when $k \geq 2$. This is due to a bad example for the case $k=3$ that shows an existence of a quasi-bipartite graph that cannot be decomposed into {\em $k$-divergent Steiner trees}. To be formal, the $k$-divergent Steiner tree is a collection of $k$ trees such that whenever we fix one terminal and pick $r,t$-paths, one from each tree, these $k$ paths are edge-disjoint. Such a collection of trees does not exist for $k\geq 3$ even in quasi-bipartite-graphs. (Please see more details in \Cref{app:bad-example}.) Thus, we completely rule out the possibility of using this approach.
    
    \item {\bf Non-Uncrossable Families of Deficient Sets.} Secondly, the Halo-Set decomposition method does not work directly for us. This is because the previous applications of the Halo-Set decomposition requires the families of deficient sets to be uncrossable. That is, one must be able to decompose the deficient sets  (i.e., a subset of vertices that needs at least one incoming edge to satisfy the connectivity requirement) into families, in which any two members from different families are disjoint. This is not the case for us, and the absent of this property has been the biggest obstacle in obtaining any non-trivial result for $k$-DST. Although, as we will show in \Cref{sec:structures}, there are many structures that resemblance those in the previous works, we have to proceed with uncrossable families of deficient sets.
    
    \item {\bf Combinatorial Greedy Algorithm is Not Available.} Thirdly, the previous application of the spider decomposition method \cite{KleinR95,Nutov10} requires the decomposition of an optimal solution into a collection of spiders \cite{KleinR95} or stars \cite{Nutov10}. Shortly, the spider decomposition method decomposes an optimal solution into a collection of spiders (resp., stars), which defines an instance of the Set Cover problem. Thus, an application of the ``combinatorial greedy algorithm'' for the set cover problem almost immediately gives a factor $O(\log n)$ approximation algorithm for the vertex-weighted Steiner tree problem \cite{KleinR95} and a factor $O(k\log n)$ approximation algorithm for the Minimum Power Cover problems in \cite{Nutov10}.
    
    The combinatorial greedy algorithm has an advantage that even though an instance of the Set Cover problem has exponential number of sets. It can run on a compact representation of an instance; see, e.g., \cite{Nutov10}.
    However, as we will discuss later, our algorithm is based on the {\em connectivity augmentation} framework, which requires an LP-based approximation algorithm. Thus, we need to decompose a ``fractional'' optimal solution for $k$-DST, which introduces some complication into our proof (even though we tried our best to keep the proof simple).
    
    \item {\bf Spider Decomposition Consisting of Disconnected Components.} In addition, while our technique is a generalization of  the spider decomposition method, each component we have to deal with (which is supposed to be a spider) is not connected and may contain directed cycles. This causes a slight complication and makes our decomposition departs from the previous two applications of the spider decomposition method \cite{KleinR95,Nutov10}.
\end{itemize}

\subsection{Our Result}
\label{sec:intro:result}

The main result in our paper is an $O(\log q\log k)$-approximation algorithm for $k$-DST on quasi-bipartite graphs. Since our algorithm is LP-based, it also gives an upper bound on the integrality gap of the standard LP.

\begin{theorem}
\label{thm:main-theorem}
Consider the $k$-Connected Directed Steiner Tree problem where an input graph consists of an $n$-vertex quasi-bipartite graph and a set of $q$ terminals.
There exists a randomized polynomial-time $O(\log q\log k)$-approximation algorithm. 
Moreover, the algorithm gives an upper bound on the integrality gap of $O(\log q\log k)$ for the standard cut-based LP-relaxation of the problem.
\end{theorem}

We also present a derandomization of our algorithm using the method of conditional expectation, which preserves the performance guarantee.
\begin{theorem}
\label{thm:main-theorem-det}
Consider the $k$-Connected Directed Steiner Tree problem where an input graph consists of an $n$-vertex quasi-bipartite graph and a set of $q$ terminals.
There exists a {\em deterministic} polynomial-time $O(\log q\log k)$-approximation algorithm.
\end{theorem}

\section{Preliminaries}
\label{sec:prelim}

We use standard graph terminologies. Given a graph $G$, we denote by $V(G)$ and $E(G)$ the vertex set and the edge set of $G$, respectively.
For any subset of vertices $U\subseteq V(G)$, we denote by $\delta^{in}_G(U)$ the set of edges in $G$ entering the set $U$ and denote by $\deg^{in}_G(U)$ its cardinality. We denote by $E_G(U)$ the set of edges that have both head and tail in $U$. That is, 
\begin{align*}
\delta^{in}_G(U) &= \{wv\in E(G): v\in U, w\not\in U\}, 
\quad \deg^{in}_G(U) = |\delta^{in}_G(U)|,\quad \text{and}\\
E_G(U) &= \{vw\in E(G): v,w\in U\}.
\end{align*}
We will omit the subscript $G$ if the graph $G$ is known in the context, and we may replace $E_G$ with another edge-set, e.g., $E_{+}$.
For any subset of edges $E'$, we denote the total cost of edges in $E'$ by $\cost(E')=\sum_{e\in E'}c_e$.

%


\subsection{Problem Definitions}
\label{sec:prelim:problem}

\paragraph*{$k$-Edge-Connected Directed Steiner Tree ($k$-DST).}
In the {\em $k$-Edge-Connected Directed Steiner Tree} problem (k-DST), we are given a graph $G$ with non-negative edge-costs $c : E\longrightarrow \mathbb{R}_0^+$, a root vertex $r$ and a set of $q$ terminals $T\subseteq (V(G)\setminus\{r\})$, and the goal is to find a minimum-cost subgraph $H\subseteq G$ such that $H$ has $k$ edge-disjoint $r\rightarrow{t}$-paths for every terminal $t\in T$. 

\paragraph*{Rooted Connectivity Augmentation (Rooted-Aug).}
In Rooted-Aug, we are given a graph $G$ with the edge-set $E(G)=E_0\cup E_+$, where $E_0$ is the set of zero-cost edges and $E_+$ is the set of positive-cost edges, a root vertex $r$ and a set of terminals $T\subseteq V(G)\setminus{r}$ such that $E_0$ induces a subgraph $G_0\subseteq G$ that has $\ell$ edge-disjoint $r\rightarrow{t}$-paths for every terminal $t\in T$. The goal in this problem is to find a minimum-cost subset of edges $E'\subseteq E_+$ such that $E_0\cup E'$ induces a subgraph $H\subseteq G$ that has $\ell+1$ edge-disjoint $r\rightarrow{t}$-paths for every terminal $t\in T$.

We may phrase Rooted-Aug as a problem of covering deficient sets as follows.
We say that a subset of vertices $U\subseteq V(G)$ is a {\em deficient set} if $U$ separates the root vertex $r$ and some terminal $t\in T$, but $U$ has less than $\ell+1$ incoming edges (which means that $U$ has exactly $\ell$ incoming edges); that is, $U$ is a deficient set if $r\not\in U$, $U\cap T\not\emptyset$ and $\deg^{in}_{G_0}(U) = \ell$. These subsets of vertices need at least one incoming edge to satisfy the connectivity requirement.
We say that an edge $e\in E_+$ {\em covers} a deficient set $U$ if $\deg^{in}_{E_0\cup \{e\}}(U) \geq \ell$, which means that adding $e$ to $G_0$ satisfies the connectivity requirement on $U$.

Let $\mathcal{F}$ denote the set of all deficient sets in the graph $G_0$.
Then Rooted-Aug may be phrased as the problem of finding a minimum-cost subset of edges $E'\subseteq E_+$ that covers all the deficient sets, which can be described by the following optimization problem:
$$
\min\{E'\subseteq E_+:\deg^{in}_{E'}(U)\geq 1\ \forall U\in\mathcal{F}\}.
$$

\paragraph*{Set Cover.}
Given a universe $\mathcal{U}$ of $n$ elements and a collection of $m$ subsets $S_1,\ldots,S_m\subseteq\mathcal{U}$, each associated with weight $w_j$, for $j=1,\ldots,m$, the goal in the Set Cover problem is to find a collection $\mathcal{S}^*$ of subsets with minimum total weights so that the union of all subsets in $\mathcal{S}^*$ is equal to $\mathcal{U}$.

\subsection{Deficient Sets, Cores and Halo-families}
\label{sec:prelim:deficient-sets}

This section discusses subsets of vertices called deficient sets that certify that the current solution subgraph in Rooted-Aug (and also in $k$-DST) does not meet the connectivity requirement.
To be formal, a subset of vertices $U\subseteq V(G)$ is called a {\em deficient set} in the graph $G$ if $T\cap U\neq\emptyset$, $r\not\in U$ and $\deg^{in}_G(U)<k$; that is, $(V(G)\setminus U,U)$ induces an edge-cut of size $<k$ that separates some terminal $t\in U\cap T$ from the root vertex $r$.
We say that an edge $vw\not\in E(G)$ {\em covers} a deficient set $U$ if $\deg^{in}_{G+vw}(U)\geq k$, i.e., the set $U$ is not a deficient set after adding the edge $vw$.
Similarly, we say that a subset of edges $E'$ covers a deficient set or a collection of deficient sets $\mathcal{F}$ if $\deg^{in}_{G+E'}(U) \geq k$, for every deficient set $U\in\mathcal{F}$.

Let $\mathcal{F}$ be a family of deficient sets.
A {\em core} $C\in \mathcal{F}$ is a deficient set such that there is no deficient set in $\mathcal{F}$ properly contained in $C$.
The {\em Halo-family} $\Halo(C)$ of a core $C$ is a collection of deficient sets in $\mathcal{F}$ that contain $C$ and contains no other core $C'\neq C$.
The {\em Halo-set} of $C$ is the union of all the deficient sets in $\Halo(C)$, i.e., $H(C) = \bigcup_{U\in\Halo(C)}U$.

\subsection{LP-relaxations}
\label{sec:prelim:LP}
 
 Throughout this paper, we will use the following standard (cut-based) LP-relaxation for $k$-DST and the Rooted-Aug. Our LP-relaxations will be written in terms of deficient sets.
 We denote by $\Value(z)$ the cost of the optimal solution to an LP $z$. 
 
 \paragraph*{LP for $k$-DST:}

Here we present the standard cut-based LP-relaxation for $k$-DST, denoted by $\LP(k)$. The collection of deficient sets in this LP is defined by $\mathcal{F}(k) = \{U\subseteq V\setminus\{r\}: U\cap T\neq\emptyset\}$.

$$
\LP(k) = \left\{
\begin{array}{lll}
\min & \sum_{e\in E}c_ex_e\\
\text{s.t.} \\
 & \sum_{e\in\delta^{in}_G(U)}x_e \geq k & \forall U\in\mathcal{F}(k)\\
 & 0 \leq x_e \leq 1 & \forall e\in E(G)
\end{array}
\right.
$$

\paragraph*{LP for Rooted-Connectivity Augmentation:}

Here we assume that the initial graph $G_0$ is already $\ell$-rooted-connected, and the goal is to add edges to increase the connectivity of the solution subgraph by one. Thus, the collection of deficient sets in this problem is defined by 
$\mathcal{F}(\ell) = \{U\subseteq V: U\cap T\neq\emptyset, \deg_{G_0}^{in}(U)| = \ell\}$.
Below is the standard cut-based LP-relaxation for the problem of increasing the rooted-connectivity of a graph by one.

\[
\LP^{aug}(\ell) = \left\{
\begin{array}{lll}
\min & \sum_{E(G)\setminus E(G_0)}c_ex_e\\
\text{s.t.} \\
 & \sum_{e\in\delta^{in}_{E(G)\setminus E(G_0)}(U)}x_e \geq 1 & \forall U\in\mathcal{F}(\ell)\\
 & 0 \leq x_e \leq 1 & \forall e\in E(G)\setminus E(G_0)
\end{array}
\right.
\]

\section{Properties of Deficient Sets in Rooted Connectivity Augmentation}
\label{sec:structures}

This section presents the basic properties of deficient sets, cores and Halo-families in a Rooted-Aug instance, which will be used in the analysis of our algorithm. Readers who are familiar with these properties may skip this section. Similar lemmas and proofs can be seen, e.g., in \cite{CheriyanL13}. 
Our proofs are rather standard. The readers who are familiar with these properties may skip to the next section. 

The first property is the uncrossing lemma for deficient sets of Rooted-Aug.

\begin{restatable}[Uncrossing Properties]{lemma}{lemuncross}
\label{lem:uncrossing}
Consider an instance of Rooted-Aug. Let $G_0$ be a rooted $\ell$-connected graph, and let $A,B$ be deficient sets in $G_0$ that have a common terminal, i.e., $A\cap B\cap T\neq\emptyset$. Then both $A\cup B$ and $A\cap B$ are deficient sets.
\end{restatable}

\begin{proof}
We prove the lemma by using Menger's theorem and the submodularity of the indegree function $\deg^{in}$.
First, since $G_0$ is rooted $\ell$-connected, we know from Menger's Theorem that $\deg^{in}(A)$ and $\deg^{in}(B)= \ell$. We also know that  $\deg^{in}(A\cup B)\geq \ell$ and $\deg^{in}(A\cap B)\geq \ell$ because the root $r$ is not contained in either $A$ or $B$ and that $A\cap B\cap T\neq \emptyset$.
By the submodularity of $\deg^{in}$, it holds that 
$$2\ell=\deg^{in}(A) + \deg^{in}(B) \geq \deg^{in}(A\cup B) + \deg^{in}(A\cap B)\geq 2\ell.$$
Therefore,  $\deg^{in}(A\cup B) = \deg^{in}(A\cap B)= \ell$, implying that both $A\cup B$ and $A\cap B$ are deficient sets in the Rooted-Aug instance.
\end{proof}

The next lemma gives an important property of the cores arose from deficient sets in directed graphs; that is, two cores may have non-empty intersection on Steiner vertices, but they are disjoint on terminal vertices.

\begin{restatable}[Members of Two Halo-families are Terminal Disjoint]{lemma}{lemterminaldisjoint}
\label{lem:terminal-disjoint}
Let $C$ and $C'$ be two distinct cores. Then, for any deficient sets $U\in \Halo(C)$ and $U'\in\Halo(C)$, it holds that $U\cap U'\cap T=\emptyset$, i.e., any members of two distinct Halo-families have no common terminals.
\end{restatable}
\begin{proof}
We prove the lemma by contradiction.
Let $U$ and $U'$ be deficient sets $U\in \Halo(C)$ and $U'\in\Halo(C)$ such that $U$ and $U'$ share a terminal $t \in U\cap U' \cap T$. 
We may assume that $U$ and $U'$ are minimal such sets, i.e., there are no deficient sets $W\in \Halo(C)$ and $W'\in\Halo(C)$ such that (1) $W$ is properly contained in $U$, (2) $W'$ is properly contained in $U'$ and (3) $t\in W\cap W'$.
By \Cref{lem:uncrossing}, $U\cap U'$ must be a deficient set properly contained in both $U$ and $U'$ (because $C\neq C')$.
This contradicts the minimality of $U$ and $U'$.
\end{proof}

The next lemma shows that both the union and the intersection of any two deficient sets in a Halo-family $\Halo(C)$ are also deficient sets in $\Halo(C)$. This is a crucial property for computing the halo-set $H(C)$ as we are unable to list all the deficient sets in a Halo-family.

\begin{restatable}[Union and Intersection of Halo-Family Members]{lemma}{lemunionisdeficient}
\label{lem:union-is-deficient}
Let $\mathcal{F}$ be a family of all deficient sets in $G_0$,
and let $C$ be any core w.r.t. $\mathcal{F}$.
Then, for any two deficient sets $A,B\in\Halo(C)$, both $A\cap B$ and $A\cup B$ are also deficient sets in $\Halo(C)$.
\end{restatable}
\begin{proof}
Consider any deficient sets $A,B\in\Halo(C)$.
Since both $A$ and $B$ contain $C$, they share at least one terminal.
Thus, \Cref{lem:uncrossing} implies that both $A\cup B$ and $A\cap B$ are deficient sets. 
Clearly, $A\cap B$ contains $C$ and no other core $C'\neq C$.
Thus, $A\cap B$ is a member of $\Halo(C)$.

Next consider $A\cup B$.
Assume for a contradiction that $A\cup B$ is not a member of $\Halo(C)$.
Then $A\cup B$ must contain a core $C'\neq C$.
This means that at least one of the sets, say $A$, contains some terminal $t\in C'$.
By \Cref{lem:uncrossing}, since $A$ and $C'$ have a common terminal, it holds that $A\cap C'$ is a deficient set.
Since $C'\subsetneq A$ (because $A$ is a member of $\Halo(C)$), we have that $A\cap C'$ is a deficient set that is strictly contained in $C'$, a contradiction.
\end{proof}

It follows as a corollary that $H(C)=\bigcup_{U\in\Halo(C)}U$ is a also deficient set in $\Halo(C)$.

\begin{corollary}[Halo-set is deficient]
\label{cor:haloset-is-deficient}
Let $\mathcal{F}$ be a family of all deficient sets in $G_0$,
and let $C$ be any core w.r.t. $\mathcal{F}$.
Then the Halo-set $H(C)=\bigcup_{U\in\Halo(C)}U$ is also a deficient set in $\Halo(C)$.
\end{corollary}

\Cref{cor:haloset-is-deficient} implies that $H(C)$ can be computed in polynomial-time using an efficient maximum-flow algorithm. Such an algorithm can be seen in \cite{CheriyanL13}.

\begin{corollary}
For any core $C$, its Halo-set $H(C)=\bigcup_{U\in\mathcal{\Halo(C)}}U$ of a core $C$ can be computed in polynomial-time.
\end{corollary}

\section{Our Algorithm and Its Overview}
\label{sec:overview}

This section provides the overview of our algorithm, which is based on the {\em connectivity augmentation} framework plus the Halo-set decomposition method. To be specific, our algorithm starts with an empty graph called $H_0=(V,\emptyset)$. Then we add edges from $G$ to the graph $H_0$ to form a graph $H_1$ that has at least one path from the root vertex $r$ to each terminal $t\in T$. We keep repeating the process, which produces graphs $H_2, \ldots, H_k$ such that each graph $H_{\ell}$, for $\ell\in[k]$, has $\ell$ edge-disjoint $r,t$-paths for every terminal $t\in S$. In each iteration $\ell\in[k]$, we increase the rooted-connectivity of a graph by one using the Halo-set decomposition method.

We discuss the connectivity augmentation framework in \Cref{sec:overview:augmentation} and discuss the algorithm based on the Halo-set decomposition method for Rooted-Aug in \Cref{sec:overview:halo-set-method}. We devote \Cref{sec:covering-halo-families} to present a key subroutine for solving the the problem of {\em covering Halo-families} via a reduction to the Set Cover problem.

\subsection{Connectivity Augmentation Framework}
\label{sec:overview:augmentation}

A straightforward analysis of the connectivity augmentation framework incurs a factor $k$ in the approximation ratio.
Nevertheless, provided that the approximation algorithm for Rooted-Aug is based on the standard LP for $k$-DST, the cost incurred by this framework is only $\sum_{\ell=1}^{k}1/(k-\ell+1)=O(\log{k})$. This is known as the {\em LP-scaling technique}, which has been used many times in literature; see, e.g., \cite{GoemansGPSTW94,KortsarzN05,CheriyanLNV14}. 

\begin{restatable}[LP-Scaling]{lemma}{lemLPscaling}
\label{lem:LP-scaling}
Consider an instance of the $k$-DST problem, and its corresponding LP, namely $\LP(k)$.
Suppose there exists an algorithm that produces an integer solution to $\LP^{aug}(\ell)$ with costs at most $\alpha_{\ell}\cdot \Value(\LP^{aug}(\ell))$.
Then there exists an $\sum_{\ell=1}^k\alpha_{\ell}/(k-\ell+1)=O(\alpha\log{k})$ approximation algorithm for $k$-DST, where $\alpha=\max_{\ell=1}^k\alpha_{\ell}$.
\end{restatable}
\begin{proof}
Let $G$ be the input graph in the $k$-DST instance. 
Let $H^*$ be an optimal integral solution to $k$-DST (and thus $\LP(k)$), and let $G_0\subseteq G$ be the initial solution subgraph of Rooted-Aug where we wish to increase the connectivity of $G_0$ from $\ell$ to $\ell+1$ by adding edges from $E(G)\setminus E(G_0)$. Then we can define the following LP solution $\{x_e\}_{e\in E(G)}$ to $\LP^{aug}(\ell)$:
$$
x_e = \left\{
\begin{array}{ll}
\frac{1}{k-\ell} & \text{ if $e\in E(H^*)\setminus E(G_0)$} \\ 
0                  & \text{otherwise}.
\end{array}
\right.
$$
Let $\mathcal{F}$ be the family of deficient sets in the Rooted-Aug instance. Then we know by Menger's theorem that any deficient set $U\in\mathcal{F}$ has at least $k$ incoming edges in $H^*$, and at most $\ell$ of them are in $G_0$ (because $\deg_{G_0}^{in}(U)=\ell$ by the definition of the deficient set). Consequently, we have 
$$
\sum_{e\in\delta_{E(G)\setminus E(G_0)}^{in}(U)}x_e \geq (k-\ell)\cdot \frac{1}{k-\ell} = 1.
$$
This means that $\{x_e\}_{e\in E(G)\setminus E(G_0)}$ is a feasible solution to $\LP^{aug}(\ell)$ whose cost is at most $(1/(k-\ell))\Value(\LP(k))$. The lemma then follows by taking the summation over all $\ell=0,1,\ldots,k-1$.
\end{proof}

\subsection{Algorithm for Rooted-Aug via Halo-set Decomposition}
\label{sec:overview:halo-set-method}

The algorithm for rooted-connectivity augmentation is built on the Halo-set Decomposition framework. In detail, we decompose vertices in the graph $G_0$ into a collection of subsets of vertices, each is defined by a Halo-family $\Halo(C)$, which is in turn defined by its core $C$. Then we add edges to cover all the deficients that are contained in any of these families. However, the collection of Halo-families does not include all the deficient sets in the graph because a deficient that contain two distinct cores are not recorgnized by any Halo-families. Thus, after we cover all these Halo-families (i.e., we add edges covering all its members), we need to recompute the deficient sets remaining in the graph and form the system of Halo-families again.

Following the above method, our algorithm runs in multiple iterations. In each iteration, we first compute all the cores and ther corresponding Halo-set in the current solution subgraph, which can be done in polynomial time. (We recall that it is not possible to compute a Halo-family explicitly because it may contain exponential number of deficient sets.) These cores define a collection of Halo-families. Our goal is then to find a subset of edges $E'$ that covers Halo-families in this collection. To be formal, by {\em covering a Halo-family}, we mean that we find a subset of edges that covers every deficient set in its family.
Here our algorithm departs from the previous application of the Halo-set decomposition as we are not aiming to cover all the Halo-families. 
We cover only a constant fraction of Halo-families from the collection, which is sufficient for our purposes. 
Once we found the subset of edges $E'$, we add it to the solution subgraph and recompute the cores and their Halo-sets. 

To find a set of edges $E'$, we need to compute an optimal solution  to the LP for augmentation the connectivity of a graph from $\ell$ to $\ell+1$ (i.e., $\LP^{aug}(\ell)$), denoted by $\{x_e\}_{e\in E_+}$, where $E_+$ is the set of edges not in initial solution subgraph $H_{\ell}$, which is $\ell$-rooted-connected.
Using this LP-solution, we can find a set of edges $E'$ that covers at least $1/9$ fraction of the collection of Halo-families whose cost is at most $4\sum_{e\in E_+}c_ex_e$ via a reduction to the Set Cover problem. This subroutine is presented in \Cref{sec:covering-halo-families}.
Note that the mentioned subroutine is a randomized algorithm that has a constant success probabilty; thus, we may need to run the algorithm for $O(\log n)$ times to guarantee that it successes with high probability. The derandomization of our subroutine is presented in \Cref{sec:derandomize}.
Our algorithm for the rooted-connectivity augmentation is presented in \Cref{algo-1}.

\begin{algorithm}[!ht]
	\begin{algorithmic}[1]
		\caption{Rooted-Connectivity Augmentation}\label{algo-1}
		\Require: An input graph $G$ and an $\ell$-rooted-connected graph $H_{\ell}$ 
		\Ensure: An $(\ell+1)$-rooted-connected graph $H_{\ell+1}$
		\State Initialize $H_{\ell+1} := H_{\ell}$.
		\Repeat
		\State Find an optimal solution ${\bf x}$ to $\LP^{aug}(\ell)$.
		\State Compute cores and their corresponding Halo-sets in $H_{\ell+1}$. \label{algo-1:loop}
		\State Find a subset of edges $E'$ that covers at least $1/9$ fraction of the Halo-families whose cost is at most $4\sum_{e\in E_+}c_ex_e$. 
		\State Update $H_{\ell+1} := H_{\ell+1} + E'$. 
		\Until The graph $H_{\ell+1}$ has no deficient set (and thus has no core).
		\State \Return $H_{\ell+1}$
	\end{algorithmic}
\end{algorithm}

One may observe that the covering problem in our setting is different from that in the usual Set Cover problem as after we add edges to cover $\gamma$ fraction of the Halo-families, it is not guaranteed that the number of Halo-families will be decreased by a factor $\gamma$ as some of the deficient sets in the previous iterations may become new cores in the solution subgraph. Fortunately, we have a key property that any new core that was not contained in any Halo-families must contain at least two old cores. As a result, we can promise a factor $(1-\gamma/2)$ decrease. Please see Figure~\ref{fig:core_merge} for illustration. The subsets $C_1$ and $C_2$ are two cores covered by $e_1$ and $e_2$, respectively. After adding two edges, $C_1$ and $C_2$ are no longer a deficient set. Now the deficient set $C_3 \supseteq C_1 \bigcup C_2$ becomes a new core, which contains two old cores.
\begin{figure}[ht]
\centering
\includegraphics[scale=0.6]{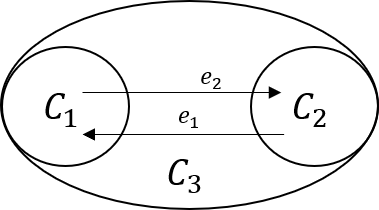}
\caption{After adding edges $e_1$ and $e_2$ to cover $C_1$, $C_2$, a new core $C_3$ appear. The new core $C_3$ must contain at least two old cores.}
\label{fig:core_merge}
\end{figure}

%
%
%

\begin{lemma}[The number of cores decreases by a constant factor] \label{lem:cores-decrease}
Let $H$ be the current solution subgraph whose number of cores is $\nu$,
and let $E'$ be a set of edges that covers at least $\gamma$ fraction of the Halo-families in $H$.
Then the number of cores in $H\cup E'$ is at most $(1-\gamma/2)\nu$.
In particular, the number of cores in the graph $H\cup E'$ decreases by a constant factor, provided that $\gamma$ is a constant.
\end{lemma}

\begin{proof}
	Let us count the number of cores in the graph $H\cup E'$.
	Consider any core $C$ in $H\cup E'$.
	If $C$ is a member of some Halo-families $\Halo(C')$ in $H$, then we know that $\Halo(C')$ is not covered by $E'$.
	Thus, there can be at most $(1-\gamma)\nu$ cores of this type.
	
	Next assume, otherwise, that $C$ is not a member of any Halo-family in $H$.
	Then, by definition, $C$ must contain at least two cores in $H$.
	Notice that, for every core $C'$ in $H$ that is contained in $C$, all of the deficients in $\Halo(C')$ must be covered by $E'$.
	Suppose not.
	Then there exists a deficient set $U$ in $\Halo(C')$ that is not covered by $E'$.
	Since $U$ interesects $C$ on the terminal set, \Cref{lem:uncrossing} implies that $U\cap C$ is also a deficient set.
	By \Cref{lem:terminal-disjoint}, any two cores are disjoint on the terminal set, which means that $U\cap C$ is strictly contained in $C$ (because $C$ contains another core $C''$ distinct from $C'$).
	The existence of $U\cap C$ contradicts the fact that $C$ is a core in $H\cup E'$.
	Thus, we conclude that $H\cup E'$ has at most $(\gamma/2)\nu$ cores of this type.

	Summing it up, the total number of cores in $H\cup E'$ is at most $(1-\gamma/2)\nu$ as claimed.
\end{proof}

It follows as a corollary that our algorithm terminates within $O(\log q)$ iterations.
\begin{corollary} \label{cor:num-iterations}
	The number of iterations of our algorithm is at most $O(\log q)$, where $q$ is the number of terminals.
\end{corollary}

By \Cref{cor:num-iterations}, our algorithm for rooted-connectivity augmentation terminates with in $O(\log q)$, and each round, we buy a set of edges whose cost is at most $4\sum_{e\in E_+}c_ex_e$; see \Cref{sec:covering-halo-families}. Therefore, the total cost incurred by our algorithm is at most $O(\log q)$ times the optimal LP solution, implying an LP-based $O(\log q)$-approximation algorithm as required by \Cref{lem:LP-scaling}. The following lemma then follows immediately.

\begin{lemma}\label{lem:rooted-aug_main}
	Consider the problem of augmenting the rooted-connectivity of a directed graph from $\ell$ to $\ell+1$ when an input graph is quasi-bipartite. 
	There exists a polynomial-time algorithm that gives a feasible solution whose cost at most $O(\log q)$ that of the optimal solution to the standard LP-relaxation.
	In particular, there exists a polynomial-time LP-based $O(\log q)$-approximation algorithm for the problem.
\end{lemma}

\paragraph*{Remark}
Lastly, we remark that one may simply cover all the Halo-families in each iteration. However, the number of rounds the randomized algorithm required will be at least $O(\log q)$, meaning that the total number of iterations is $O(\log^2 q)$.  Consequently, this implies that the algorithm has to pay a factor $O(\log^2 q)$ in the approximation ratio. We avoid the extra $O(\log q)$ factor by covering only a constant fraction of the Halo-families.

\subsection{Correctness and Overall Analysis}
\label{sec:overview:analysis}
First, to prove the feasibility of the solution subgraph, 
it suffices to show that the rooted-connectivity of the solution subgraph increasess by at least one in each connectivity augmentation step. 
This simply follows by the stopping condition of the Halo-set decomposition method that it runs until there exists no core in the graph (and thus no deficient sets). 
It then follows by Menger's theorem that the number of edge-disjoint paths from the root vertex $r$ to each terminal $t\in T$ must be increased by at least one.

Next we analyze the cost. 
By \Cref{lem:rooted-aug_main},  the approximation factor incurred by \Cref{algo-1} is $O(\log q)$, and it also bounds the integrality gap of $\LP^{aug}(\ell)$. 
Consequently, letting $\opt_k$ denote the cost of an optimal solution to $k$-DST,
by Lemma~\ref{lem:LP-scaling}, the total expected cost incurred by the algorithm is then 
\begin{align*}
\sum_{\ell=1}^{k}O(\log q)\cdot \Value(\LP^{aug}(\ell)) 
 &= O(\log q)\cdot \left(\sum_{\ell=1}^{k}\frac{1}{k-\ell+1}\right) \cdot \Value(\LP(k))\\
 &= O(\log q\log k)\cdot \opt_k.
\end{align*}
This completes the proof of \Cref{thm:main-theorem} (and also \Cref{thm:main-theorem-det}).

\section{Covering Halo-Families via Set Cover}
\label{sec:covering-halo-families}

In this section, we present our subroutine for covering the Halo-families that arose from the Rooted-Aug problem. As mentioned in the introduction, the key ingredient in our algorithm is the reduction from the problem of covering Halo-families to the Set Cover problem. However, our instance of the Set Cover problem has an exponential number of subsets, which more resemblances to an instance of the {\em Facility Location} problem. To prove our result, one route would be using {\em Facility Location} as an intermediate problem in the presentation. However, we prefer to directly apply a reduction to the Set Cover problem to avoid confusing the readers.

\subsection{The Reduction to Set Cover and Algorithm}
\label{sec:covering-halo-families:reduction}

As an overview, our reduction follows from simple observations.
\begin{itemize}
    \item[(P1)] For any minimal subset of edges that covers a Halo-family $\Halo(C)$, there is only one edge $e$ that has head in $\Halo(C)$ and tail outside. Let us say $e$ is outer-cover $\Halo(C)$ since it is coming from the outside of the family.
    \item[(P2)] Any edge can be contained in at most one $\Halo(C)$, i.e., there is at most one halo-families $\Halo(C)$ such that both head and tail of $e$ are contained in $H(C)$.
    \item[(P3)] An LP for covering a single Halo-family is integral.
\end{itemize}

Now an instance of the Set Cover problem can be easily deduced. We define each Halo-family $\Halo(C)$ as an element, and we define each edge $e$ as a subset.
However, we may have multiple subsets corresponding to the same edge $e$ as it may serve as an ``outer-cover'' for many Halo-families. Thus, we need to enumerate all the possible collections of Halo-families that are outer-covered by $e$. We avoid getting exponential number of subsets by using the solution from an LP (for the connectivity augmentation problem) as a guideline. 

Before proceeding, we need to formally define some terminologies. 
Let $\hat{G}$ be the current solution subgraph. 
We say that an edge $e$ {\em outer-covers} a Halo-family $\Halo(C)$ if the head of $e$ is in $H(C)$ and the tail is not in $H(C)$ and that there exists a subset of edges $E'\subseteq E_+\setminus E(\hat{G})$ such that (1) both endpoints of every edge in $E'$ are contained in $H(C)$ and (2) the set of edges $E'\cup\{e\}$ covers $\Halo(C)$.

For each Halo-families $\Halo(C)$, we define the set of edges $I^e_C$ to be the minimum-cost subset of edges $E'\subseteq E_+\setminus E(\hat{G})$ whose both endpoints are in $H(C)$ and that $E'\cup\{e\}$ covers $\Halo(C)$, and we denote the cost of $I^e_C$ by $\dist^e_C$. 
We may think that $\dist^e_C$ is the cost for covering $\Halo(C)$ given that $e$ has been taken for free. 
We use the notation $E[C]$ to mean the set of edges whose both endpoints are contained in the Halo-set $H(C)$.
We denote the cost of the fractional solution restricted to $E[C]$ by $\cost_x(E(C)) = \sum_{e\in E[C]}c_ex_e$.

Our reduction is as follows.
Let $H$ be the current solution subgraph.
For each core $C$ in $H$, we define an element $C$.
For each edge $e\in E_+\setminus E(H)$, we define a subset $S_e$ by adding to $S_e$ an element $C$ if $\dist^e_C \leq \cost_x(E[C])$.
This completes a reduction. It is not hard to see that the resulting instance of the Set Cover problem has polynomial size.
To show that our reduction runs in polynomial-time, we need to give a polynomial-time algorithm for computing $\dist^e_C$, which we defer to \Cref{sec:covering-halo-families:properties}. Here we leave a forward reference to \Cref{lem:dist-polytime}. 
Our algorithm that covers a constant fraction of the collection of Halo-families is then followed by simply 
picking each edge $e$ with probability $x_e$ and add all the edges $I^e_C$, for all cores $C\in S_e$, to the solution subgraph;
if a core $C$ is outer-covered by two picked edges, then we add only one edge-set $I^e_C$.
We claim that the set of edges chosen by our algorithm covers at least $1/9$ fraction of the Halo-families, while paying a cost of at most four times the optimum (with a constant probability).
In particular, we prove the following lemma.

\begin{lemma}
\label{lem:partial-cover}
With constant probability, 
the above algorithm covers at least $1/9$ fraction of the collection of Halo-families,
and the cost of the of the edges chosen by the algorithm has cost at most $4\sum_{e\in E_+}c_ex_e$.
In particular, the algorithm partially covers the collection of the Halo-families, while paying 
the cost of at most constant times the optimum.
\end{lemma}

To prove \Cref{lem:partial-cover}, we need to show that the fractional solution defined by $\{x_e\}_{e\in E+}$ is (almost) feasible to the Set Cover instance, which then implies that the set of edges we bought covers a constant fraction of the Halo-families with probability at least $2/3$.
Then we will show that the cost of the fractional solution to the Set Cover instance is at most twice that of the optimal solution to $\LP^{aug}(\ell)$,
thus implying that we pay at most six times the optimum with probability $2/3$.

To be more precise, we show in \Cref{sec:covering-halo-families:partial-cover} that our algorithm covers at least $1/3$ fraction of the Halo-families in expectation, meaning that we cover less than $1/9$ fraction with probability at most $1/3$.
Then we show in \Cref{sec:covering-halo-families:cost} that the expected cost incurred by our algorithm is $2\sum_{e\in E_{+}}c_ex_e$, thus implying that we pay more than six times that of the LP with probability at most $1/3$.
Applying the union bound, we conclude that our algorithm covers at least $1/9$ fraction of the Halo-families, while paying the cost of at most six times the optimal LP solution with probability at least $1/3$.
(Note that in \Cref{sec:covering-halo-families:cost}, we show a slightly stronger statement that the cost incurred by our algorithm is $4\sum_{e\in E_{+}}c_ex_e$ with probability at least $2/3$.)
To finish our proof, we proceed to prove the above two claims and then prove the structural properties used in the forward references.

\subsection{Partial Covering}
\label{sec:covering-halo-families:partial-cover}

We show in this section that our algorithm covers at least $1/9$ fraction of the Halo-families with probability at least $1/3$

First, we show that the LP variable defined by $x_e$ is almost feasible to the LP-relaxation of the Set Cover problem. We note that our proof will need a forward reference to \Cref{lem:minimal-fractional-cover}.
\begin{lemma}
\label{lem:set-cover-feasibility}
The LP variable $\{y_e\}_{e\in E_+}$, where $y_e =\min\{1, 2x_e\}$ for all edges $e\in E_+$ is feasible to the Set Cover instance. 
That is, for any core $C$ in the graph,
$$\sum_{e\in E_+:C \in S_e}x_e \geq 1/2.$$
\end{lemma}
\begin{proof}
Consider a core $C$, which corresponds to an element in the Set Cover instance.
We take the set of edges incident to its Halo-set $H(C)$, and 
find a minimal vectors $\{x'_e\}_{e\in E_+}$ such that $\{x'_e\}_{e\in E_+}$ fractionally covers the Halo-family $\Halo(C)$ and $x'_e\leq x_e$ for all edges $e\in E_+$.
(Note that by minimality we mean that, for any edge $e$ and any $\epsilon>0$, decreasing the value of $x'_e$ by $\epsilon$ results in an infeasible solution.)
By \Cref{lem:minimal-fractional-cover}, we have $\sum_{e\in\delta^{in}(H(C))}x'_e=1$, i.e., the total weight of the LP value of edges incoming to $H(C)$ is exactly one.

Next consider the following LP.

$$
\LP^{halo}=\left\{
\begin{array}{lll}
  \min & \sum_{e'\in E_{+}(H(C))}c_{e'}x_{e'}\\
  \text{s.t} 
     & \sum_{e'\in \delta^{in}_{E_+}(U)}x_{e'}\geq 1 &\forall U\in\Halo(C)\\
     & 0 \leq x_{e'} \leq 1 & \forall e\in E_+(H(C))\\
\end{array}\right.
$$

By \Cref{lem:uncrossing}, we know that both the intersection and union of any two deficient sets in $\Halo(C)$ are also deficient sets in $\Halo(C)$. This means that the Halo-family $\Halo(C)$ is an intersecting family. 
It then follows from the result of Frank \cite{frank1979kernel} that the above LP is {\em Totally Dual Integral}, which means that any convex point of its polytope is an integral solution (including the optimal one).
Since $\{x'_e\}_{e\in E'}$ is a feasible solution to $\LP^{halo}$, it can be written as a convex combination of integral vectors in the polytope, i.e.,
$$
{\bf x} = \sum_{i=1}^w \lambda_i{\bf z}^i, \text{ where } \sum_{i=1}^w\lambda_i = 1.
$$
Let $F_i$ be the set of edges induced by each integral vector ${\bf z}^i$ (i.e., $F_i$ is the support of ${\bf z}^i$).
Since the LP requires $H(C)$ to have at least one incoming edge, we deduce that, for each $F_i$, there exists one edge $e_i\in F_i$ entering $H(C)$.

Now we compare the cost of $\dist^{e_i}_C$ to the cost of $F_i\setminus\{e_i\}$. By minimality of $\dist^{e_i}_C$, we know that $\dist^{e_i}_C \leq \cost(F_i\setminus\{e_i\})$ for all $i=1,\ldots,w$. We recall that we add a core $C$ to the set $S_{e_i}$ only if $\dist^{e_i}_C\leq\cost_{x}(E[C])$.
Since $\cost_{x'}(E[C])$ is the convex combination of ${\bf Z}^i$, at least half of the $F_i$ (w.r.t. to the weight $\lambda_i$) must have
$\dist^{e_i}_C \leq \cost(F_i\setminus\{e_i\}) \leq \cost_{x}(E[C])$;
that is, $\sum_{i:\dist^e_C \leq \cost(F_i\setminus\{e_i\})}\lambda_i \geq 1/2$.
Therefore, we conclude that the sum of $y_{e_i}$ over all $e_i$ such that $\dist^{e_i}_C \leq \cost_{x}(E[C])$ is at least one, thus proving the lemma.
\end{proof}

We remark that we may define the Set Cover instance so that $\{x_e\}_{e\in E_+}$ is exactly a feasible solution to the LP for the Set Cover problem by using the integer decomposition as in the proof of \Cref{lem:set-cover-feasibility}. However, we choose to present it this way to keep the reduction simple. 

Now we finish the proof of our claim.
Consider a core $C$.
Note that by construction, every time we pick an edge $e$, we also add the set of edges $F_C$, for each $C\in S_e$, such that $F_C\cup\{e\}$ covers $\Halo(C)$.
Thus, the probability that the algorithm picks no edges $e$ such that $C\in S_e$ is
$$
\Pi_{e\in E_+: C\in S_e}(1 - x_e)
\leq \exp\left(-\sum_{e\in E_+}x_e\right)
\leq \exp(-1/2)
\leq \frac{2}{3}.
$$
The first inequality follows because $1-x \leq \exp(-x)$, for $0 < x \leq 1$. 
That is, the probability that the algorithm does not cover a core $C$ is at most $2/3$,
which means that the expected fraction of Halo-families covered by our algorithm is at least $1/3$.
Applying Markov's inequality, we conclude that with probability at least $2/3$ our algorithms covers at least $1/9$ fraction of the Halo-families.

Our algorithm can be derandomized using the method of conditional expectation. Please see \Cref{sec:derandomize} for details.

\subsection{Cost Analysis}
\label{sec:covering-halo-families:cost}

Now we analyze the expected cost of the edges we add to the solution subgraph.
We classify the cost incurred by our algorithm into two categories. 
The first case is the set of edges $e$ that we pick with probability $x_e$.
The expected cost of this case is $\sum_{e\in E_+}c_ex_e$.
Applying Markov's inequality, we have that with probability at least $2/3$ the cost incurred by the edges of this case is at most $3\sum_{e\in E_+}c_ex_e$.

The second case is the set of edges corresponding to each subset $S_e$ whose the edge $e$ is added to the solution.
By construction, a core $C$ is added to $S_e$ only if $\cost_x(E[C])$ is greater than $\dist^e_C$ (i.e., the cost of the set of edges $I^e_C$).
We also recall that we also add one set of edges $I^e_C$ to the solution if there are more than one edges $e$ such that $C\in S_e$ are chosen. 
As the set of edges $E[C]$ and $E[C']$ are disjoint for any two cores $C\neq C'$ (please see the forward reference to \Cref{lem:internal-edge-disjoint}), we conclude that the cost incurred by the edges of this case is at most $\sum_{e\in E_+}c_ex_e$ (regardless of the choices of the edges randomly picked in the previous step). 
Therefore, with probability at least $2/3$ the cost of edges chosen by our algorithm is at most $4\sum_{e\in E_+}c_ex_e$.

\subsection{Structural Properties of the LP solution}
\label{sec:covering-halo-families:properties}

We devote this last subsection to prove properties (P1) to (P3) and all the forward references as discussed earlier.
Property (P3) simply follows from the fact that the intersection and union of any two members of a Halo-family $\Halo(C)$ are also members of $\Halo(C)$, which means that the polytope of the problem of covering $\Halo(C)$ is integral due to the result of Frank \cite{frank1979kernel}.
Thus, we are left to prove the property (P1) and (P2) and to present a polynomial-time algorithm for computing $\dist^e_C$, which thus complete the proof that our reduction can be done in polynomial time.

First, we prove Property (P1), which allows us to reduce the instance of the problem of covering Halo-families to a Set Cover instance.

\begin{lemma}[Unique Entering Edge in Minimal Cover]
\label{lem:minimal-fractional-cover}
Consider a minimal fractional cover $x$ of a Halo-family $\Halo(C)$.
That is, $x$ is a feasible solution to $\LP^{halo}$ whose collection of deficient sets is defined by $\Halo(C)$, and decreasing the value $x_e$ of any edge $e\in E_{+}$ results in an infeasible solution. 
It holds that $\sum_{e\in\delta^{in}_{E_+}(H(C))}x_e = 1$.
Thus, for an integral solution $E'$, there is exactly one edge $e\in E'$ entering the Halo-set $H(C)$.
\end{lemma}
\begin{proof}
Assume for a contradiction that $\sum_{e\in\delta^{in}_{E_+}(H(C))}x_e > 1$.
By the minimality of $x$, for any edge $e\in\delta^{in}_{E_+}(H(C))$, there exists a deficient set $W_e\in\Halo(C)$ such that\\ $\sum_{e\in\delta^{in}_{E_+}(W_e)}x_e=1$. We choose $W_e$ to be the maximum inclusionwise such set and call it the witness set of $e$. 

Now we take two distinct witness sets $W_e$ and $W_{e'}$, for $e\neq e'$.
By \Cref{lem:union-is-deficient}, both $W_e\cap W_{e'}$ and $W_e\cup W_{e'}$ are deficient sets in $\Halo(C)$.
Let us abuse the notation of $x$.
For any subset of vertices $S\subseteq V(G)$, let $x(S)=\sum_{e\in \delta^{in}_{E_+}(S)}x_e$.
The function $x(S)$ is known to be submodular \cite{FrankJ2016graph}, meaning that
$$
2 = x(W_e) + x(W_{e'}) \geq x(W_e\cap W_{e'}) + x(W_e\cup W_{e'}) \geq 2.
$$
The last inequality follows because $\{x\}_{e\in E_+}$ fractionally covers $\Halo(C)$, which then implies that $x(W_e\cap W_{e'})=x(W_e\cup W_{e'})=1$.
But, this contradicts the choice of $W_e$ (and also $W_{e'}$) because $W_e\cup W_{e'}$ is a deficient set in $\Halo(C)$ strictly containing $W_e$ in which the conditions $x(W_e\cup W_{e'})=1$ and $e\in\delta^{in}_{E_{+}}(W_e\cup W_{e'})$ hold.
\end{proof}

Next we prove Property (P2), which allows us to upper bound the cost incurred by the main algorithm.

\begin{lemma}[Internally Edge-Disjoint]
\label{lem:internal-edge-disjoint}
Consider a quasi-bipartite graph $G$. For any edge $e\in E(G)$, there is at most one core $C\in\mathcal{C}$ such that $e\in E(H(C))$.
\end{lemma}
\begin{proof}
Consider any edge $uv\in E(G)$. Since $G$ is a quasi-bipartite graph, one of $u$ and $v$ must be a terminal. By \Cref{lem:terminal-disjoint}, we know that there can be at most one Halo-family $\Halo(C)$, for some $C\in\mathcal{C}$, whose member contains both $u$ and $v$. Hence, the lemma follows.
\end{proof}

Finally, we show that $\dist^e_C$ can be computed in polynomial time. 

\begin{lemma}
\label{lem:dist-polytime}
For any core $C\in\mathcal{C}$ and an edge $e\in E(G)$, the set of edges $I^e_C$ and, thus, its cost $\dist^e_C$ can be computed in polynomial time. Moreover, the value of $\dist^e_C$ is equal to the optimal value of the corresponding covering LP given below.
$$
\LP^{cover}=\left\{
\begin{array}{lll}
  \min & \sum_{e'\in E_{+}(H(C))}c_{e'}x_{e'}\\
  \text{s.t} 
     & \sum_{e'\in \delta^{in}_{E_+}(U)}x_{e'}\geq 1 &\forall U\in\Halo(C)\\
     & 0 \leq x_{e'} \leq 1 & \forall e\in E_+(H(C))\\
     & x_e = 1
\end{array}\right.
$$
\end{lemma}
\begin{proof}
Consider the Halo-family $\Halo(C)$.
By \Cref{lem:union-is-deficient}, the union and intersection of any deficient sets $U,W\in\Halo(C)$ are also deficient sets in $\Halo(C)$. This means that $\Halo(C)$ is an {\em intersecting family}. It is known that the standard LP for covering an intersecting family is integral (see, e.g., \cite{frank1979kernel}), which implies that we can compute $\dist^e_C$ and its corresponding set of edges $I^e_C$ in polynomial time by solving $\LP^{cover}$.

Alternatively, we may compute $\dist^e_C$ combinatorially using an efficient minimum-cost $(\ell+1)$-flow algorithm. In particular, we construct an $s^*,t^*$-flow network by setting the costs of edges in $\delta^{in}_{H_{\ell+1}}(H(C))\cup\{e\}$ to zero, adding a source $s^*$ connecting to $\ell+1$ edges entering $\Halo(C)$ (which consists of $\ell$ edges from $\delta^{in}_{H_{\ell+1}}(H(C))$ plus the edge $e$) and then picking an arbitrary terminal $t^*\in C$ as a sink. All the edges not in $E(H(C))$ except $\delta^{in}_{H_{\ell+1}}(H(C))\cup\{e\}$ are removed. 
Applying Manger's theorem, it can be seen that every $(\ell+1)$-flow in this $s^*,t^*$-flow network corresponds to a feasible solution to the covering problem with the same cost. This gives a polynomial-time algorithm for computing $\dist^e_C$ and $I^e_C$ as desired. 
\end{proof}

\section{Conclusion and Open Problems}
\label{sec:conclusion}

We have presented our $O(\log q \log k)$-approximation algorithm for $k$-DST when an input graph is quasi-bipartite. 
This is the first polylogarithmic approximation algorithm for $k$-DST for arbitrary $k$ that does not require an additional assumption on the structure of the optimal solution. In addition, our result implies that $k$-DST in quasi-bipartite graphs is equivalent to the Set Cover problem when $k=O(1)$. 

Lastly, we conclude our paper with some open problems. A straightforward question is whether there exists a non-trivial approximation algorithm for $k$-DST for $k\geq 3$ in general case or for a larger class of graphs (perhaps, in quasi-polynomial-time). Another interesting question is whether our randomized rounding technique, which consists of dependent rounds of a randomized rounding algorithm for the Set Cover problem, can be applied without connectivity augmentation. If this is possible, it will give $O(\log k)$ improvements upon the approximation ratios for approximating many problems whose the best known algorithms are based on the Halo-Set decomposition technique.

\paragraph*{Acknowledgement}
Bundit Laekhanukit is partially supported by the national 1000-youth award by the Chinese Government.

The works were initiated while all the authors were at the Institute for Theoretical Computer Science at the Shanghai University of Finance and Economics, and the work were done while Chun-Hsiang Chan and Hao-Ting Wei were in bachelor and master programs in Computer Science at the Institute of Information Science, Academia Sinica, Taipei.

\bibliography{ref}

\newpage \appendix

\section{Derandomization}
\label{sec:derandomize}
In this section, we present a derandomization of our algorithm in \Cref{sec:covering-halo-families} using the method of conditional expectation \cite{AlonS16}. 
We will mostly follow the proof presented in the work of Bertsimas and Vohra~\cite{BertsimasV98} who gave a derandomized technique for the randomized scheme for the Set Cover problem.

In more detail, first observe that the cost incurred by our algorithm comes from two parts.
The firt part is the cost of edges $e$ that we pick with probabilty $x_e$,
and the second part is the cost of edges $I^C_e$ in which the edge $e$ is chosen. 
For the second part, our algorithm guarantees that, for each core $C$, only one set of edges $I^C_e$ will be added to the solution.
Thus, by the construction of $S_e$ and \Cref{lem:internal-edge-disjoint}, the cost incurred by this part is $\sum_{e\in E_+}c_ex_e$ regardless of the choices of the edges $e$ added to the solution from the first part.

Hence, it suffices to show that there exists a deterministic algorithm that pick a set of edges $E'$ that outter-covers at least $1/3$ fraction of the Halo-families, while paying the cost at most $\sum_{e\in E_+}c_ex_e$. 

Let $\mathcal{C}$ be the collection of all the cores in the current solution subgraph.
For a given set of edges $E'\subseteq E_+$, 
we define a function $\tau_C\in\{0,1\}$ for each Halo-family $\Halo(C)$ to indicate whether $\Halo(C)$ is covered by some edge in $E'$,
and we define a function $\mathbb{I}(\vec{\tau})$ to indicate whether $E'$ outer-covers at least $1/9$ fraction of the Halo-families.
The formal definition of these two functions are given below.
$$
\tau_C(E')=\left\{
\begin{array}{ll}
1 & \text{if $E'$ outer-covers $\Halo(C)$}\\
0 & \text{Otherwise}\\
\end{array}
\right.
$$

$$
\mathbb{I}(E')=\left\{
\begin{array}{ll}
1 & \sum_{C\in\mathcal{C}}\tau_C(E') < \frac{|\mathcal{C}|}{9} \\
0 & \text{Otherwise} \\
\end{array}
\right.
$$

Next we define the potential function:
$$\Phi(E') = \sum_{e\in E'}c_e + M\cdot\mathbb{I}(E') \text{, where $M=3\sum_{e\in E_+}c_ex_e$.} $$

Observe that $\Phi(E') \leq M$ if $E'$ outer-covers at least $1/9$ fraction of the Halo-families, while having the cost at most three times that of the LP solution; otherwise,  $\Phi(X) > M$.
Notice that, by \Cref{lem:partial-cover}, if we add each edge $e\in E_+$ to $E'$ with probability $x_e$, then $\E[\Phi(E')]\leq M$. 
Thus, there exists an event that $\Phi(X) \leq M$, which will give us the desired integer solution. 
We then follow the method of conditional expectation (see, e.g., \cite{AlonS16}).
That is, we order edges in $E_+$ in an arbitrary order, say $e_1,e_2,\ldots,e_{|E_+|}$.
Let $E''$ be the set of edges that we try to simulate the set of randomly chosen edges $E'$.
Initially, $E_{det}=\emptyset$.
Then we decide to add each edge $e_i$, for $i=1,2,\ldots,|E_+|$ to $E'$ if 
$\E[\Phi(E')| E_{det}\cup\{e_i\}\subseteq E']] \leq \E[\Phi(E')| E_{det}\subseteq E']] $.
This way the resulting set of edges $E_{det}$ outer-covers at least $1/9$ fraction of the Halo-families, while having the cost of at most $3\sum_{e\in E_+}c_ex_e$.
Therefore, after adding the set of edges $I^C_e$ for each core outer-covered by some edge $e\in E_{det}$, we have a set of edges that covers at least $1/9$ fraction of the Halo-families with cost at most $4\sum_{e\in E_+}c_ex_e$, i.e., with the same guarantee as desired in \Cref{lem:partial-cover}.

\section{Bad Example for Grandoni-Laekhanukit Tree-Embedding Approach}
\label{app:bad-example}

In \cite{GrandoniL17}, Grandoni and Laekhanukit proposed an approximation scheme for $k$-DST based on the decomposition of an optimal solution into $k$ divergent arborescences \cite{GeorgiadisT16,BercziK11note}. Their approach results in the first non-trivial approximation algorithm for $2$-DST, and the algorithm achieves polylogarithmic approximation ratio in quasi-polynomial-time. Nevertheless, this technique meets a barrier as soon as $k\geq 3$ as it was shown in \cite{BercziK11note} that the decomposition of an optimal solution into $k$ divergent arborescences does not exist for general graphs when $k\geq 3$. One would hope that the decomposition is still possible for some classes of graphs, e.g., quasi-bipartite graphs. We show that, unfortunately, even for the class of quasi-bipartite graphs the divergent arborescences decomposition does not exist for $k\geq 3$. The counter example of a $3$-rooted-connected graph that has no $3$ divergent arborescences is shown in \Cref{fig:bad-example}.

\begin{figure}
    \centering
    \includegraphics[scale=0.6]{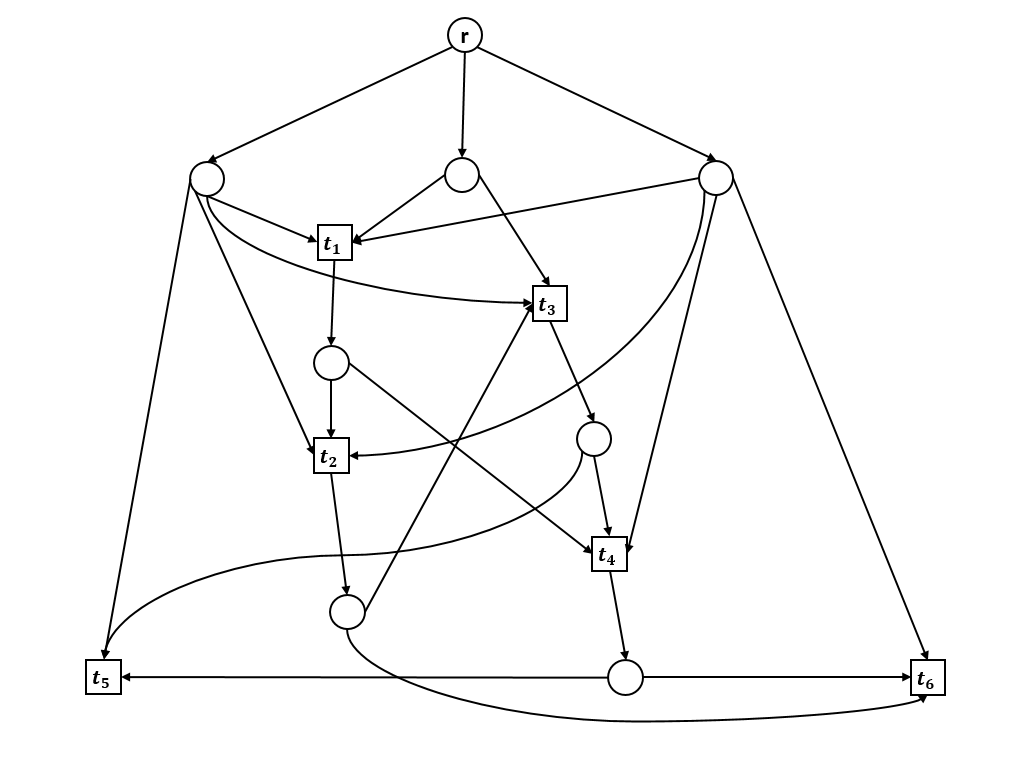}
    \caption{This figure shows an example $3$-rooted-connected quasi-bipartite graph that cannot be decomposed into $3$ divergent arborescences.}
    \label{fig:bad-example}
\end{figure}

\end{document}